\documentclass[runningheads]{llncs}
\usepackage{graphicx}
\usepackage{amssymb}

\begin{document}
\title{A Polynomial Time Algorithm for Fair Resource Allocation in Resource Exchange\thanks{Supported by the National Nature Science Foundation of China (No. 11301475, 61632017, 61761146005). A preliminary version is accepted by FAW 2019 \cite{yan2019polynomial}. Several omitted proofs are presented in the Appendix.}}
\titlerunning{Polynomial Algorithm for Fair Resource Allocation}
\author{Xiang Yan\inst{1} \and
Wei Zhu\inst{2}}
\authorrunning{Yan and Zhu}
\institute{Shanghai Jiao Tong University, Shanghai 200240, China \email{xyansjtu@163.com} \and
China Academy of Aerospace Standardization and Product Assurance, Beijing 100071, China \email{shakerswei@sina.com}}
\maketitle              % typeset the header of the contribution
\begin{abstract}
The rapid growth of wireless and mobile Internet has led to wide applications of exchanging resources over network, in which how to fairly allocate resources has become a critical challenge. 
% To motive sharing, Wu and Zhang~\cite{WZ} introduced a proportional response protocol for fairness consideration from a combinatorial structure, called \emph{bottleneck decomposition}. 
% They also showed its economic efficiency by the reduction from bottleneck decomposition to a market equilibrium. 
% From the perspective of mechanism design, Cheng \emph{et al}~\cite{CDPY,CDQY} named this protocol as BD Mechanism and proved its truthfulness against two kinds of strategic behaviors.  
To motivate sharing, a \emph{BD Mechanism} is proposed for resource allocation, which is based on a combinatorial structure called \emph{bottleneck decomposition}.
The mechanism has been shown with properties of fairness, economic efficiency~\cite{WZ}, and truthfulness against two kinds of strategic behaviors~\cite{CDPY,CDQY}.
Unfortunately, the crux on how to compute a bottleneck decomposition of any graph is remain untouched. 
In this paper, we focus on the computation of bottleneck decomposition to fill the blanks and prove that the bottleneck decomposition of a network $G=(V,E;w_v)$ can be computed in $O(n^6\log(nU))$, where $n=|V|$ and $U=max_{v\in V}w_v$. 
Based on the bottleneck decomposition, a fair allocation in resource exchange system can be obtained in polynomial time. 
% In addition, our work completes the computation of a market equilibrium and the design of BD mechanism in resource exchange system. 
% Finally we build a connection between two concepts of fairness of allocation protocol.
In addition, our work completes the computation of a market equilibrium and its relationship to two concepts of fairness in resource exchange.

\keywords{Polynomial Algorithm \and Fair Allocation \and Resource Exchange \and Bottleneck Decomposition}
\end{abstract}
\section{Introduction}\label{sec1}
The Internet era has witnessed plenty of implementations of resource exchange \cite{G,H,N,O,Swap}.
It embodies the essence of the \emph{sharing economy} and captures the ideas of collaborative consumption of resource (such as the bandwidth) to networks with participants (or agents)~\cite{FS}, such that agents can benefit from exchanging each own idle resource with others.
In this paper, we study the resource exchange problem over networks, which goes beyond the peer-to-peer
(P2P) bandwidth sharing idea~\cite{WZ}. 
Peers in such networks act as both suppliers and customers of resources, and make their resources directly available to other network peers according to preset network rules \cite{S2001}. 

The resource exchange problem can be formally modeled on an undirected connected graph $G=(V,E;w)$, where
each vertex $u\in V$ represents an agent with $w_u$ units of divisible idle resources (or weight) to be distributed among its neighbor set $\Gamma(u)$. 
The utility $U_u$ is determined by the total amount of the resources obtained from its neighbors. 
Define $x_{uv}$ to be the fraction of resource that agent $u$ allocates to $v$ and call the collection $X=(x_{uv})$ an \emph{allocation}. 
Then the utility of agent $u$ under allocation $X$ is $U_u=\sum_{v\in \Gamma(u)} x_{vu}w_v$, subject to the constraint of $\sum_{v\in \Gamma(u)}x_{uv}\leq 1$.

One critical issues for the resource exchange problem is how to design a resource exchange protocol to maintain agents' participation in a fair fashion. 
Ideally, the resource each agent obtains can compensate its contribution, but such a state may not exists due to the structure of the underlying networks.
Thus Georgiadis \emph{et al}.~\cite{GIT} thought that an allocation is fair if it can balance the exchange among all agents as much as possible. 
An \emph{exchange ratio} of each agent is defined then, to quantify the utility it receives per unit of resource it delivers out, i.e. $\beta_u(X)=\frac{U_u}{\sum_{v\in \Gamma(u)}x_{uv}w_u}$ for given allocation $X$.
In~\cite{GIT}, an allocation $X$ is said to be fair, if its exchange ratio vector $\beta(X)=(\beta_u(X))_{u\in V}$ is lexicographic optimal (lex-optimal for short). 
And a polynomial-time algorithm is designed to find such a fair allocation by transforming it to a linear programming problem.

On the other hand, Wu and Zhang~\cite{WZ} pioneered the concept of ``proportional response" inspired by the idea of ``tit-for-tat" for the consideration of fairness. 
Under a proportional response allocation, each agent responses to the neighbors who offer resource to it by allocating its resource in proportion to how much it receives. 
Formally, the proportional response allocation is specified by $x_{uv}=\frac{x_{vu}w_v}{\sum_{k\in \Gamma(u)}x_{ku}w_k}$. 
The authors showed the equivalence between such a fair allocation and the \emph{market equilibrium} of a pure exchange economy in which each agent sells its own resource and uses the money earned through trading to buy its neighbors' resource.

The algorithm to get a fair allocation in \cite{WZ} includes two parts: computing a combinatorial decomposition, called \emph{bottleneck decomposition}, of a given graph; and constructing a market equilibrium from the bottleneck decomposition. 
The work in \cite{WZ} only involved the latter. 
But how to compute the bottleneck decomposition is remained untouched. 
In this paper, we design a polynomial time algorithm to address the computation of the bottleneck decomposition.
In addition, we show the equilibrium allocation (i.e. the allocation of the market equilibrium) from the bottleneck decomposition also is lex-optimal. 
Such a result establishes the connection between the two concepts of fairness in~\cite{GIT} and~\cite{WZ}.

Another contribution of this work is to complete the computation of the market equilibrium in a special setting of a linear exchange market. 
The study of market equilibrium has a long and distinguished history in economics, starting with Arrow and Deberu's solution~\cite{AD} which proves the existence of the market equilibrium under mild conditions. 
Much work has focused on the computational aspects of the market equilibrium. 
Specially for the linear exchange model, Eaves \cite{E} first presented an exact algorithm by reducing to a linear complementary problem.
Then Garg \emph{et al}.~\cite{GMSV} derived the first polynomial time algorithm through a combinatorial interpretation and based on the characterization of equilibria as the solution set of a convex program. 
Later Ye~\cite{Ye} showed that a market equilibrium can be computed in $O(n^8\log^2 U)$-time with the interior point method. 
Recently, Duan \emph{et al.}~\cite{DGM} improved the running time to $O(n^7\log^2(nU))$ by a combinatorial algorithm. 
Compared with the general linear exchange model, we further assume that the resource of any agent is treated with equal preference by all his neighbors.
Based on it, Wu and Zhang~\cite{WZ} proposed the bottleneck decomposition, which decomposes participants in the market into several components, and showed that in a market equilibrium, trading only happens within each component. 
Therefore, the problem of computing a market equilibrium is reduced to one of computing the bottleneck decomposition. 
Our main task in this paper is to design a polynomial time algorithm to compute the bottleneck decomposition and to complete the computation of a market equilibrium in~\cite{WZ}. 
The time complexity of our algorithm is $O(n^6\log(nU))$, which is better than the algorithm of Duan \emph{et al.}~\cite{DGM}, because of further assumption in our setting.

In the rest of this paper, we introduce the concepts and properties of bottleneck decomposition in Section 2 and provide a polynomial time algorithm for computing the bottleneck decomposition in Section 3. 
In Section 4 we describe a market equilibrium from the bottleneck decomposition and show the fairness of the equilibrium allocation. At last we conclude this paper in Section 5.

\section{Preliminary}\label{sec2}

We consider a resource exchange problem modeled on an undirected and connected graph $G=(V,E;w)$ with vertex set $V$ and edge set $E$, respectively, and $w: V\rightarrow R^{+}$ is the weight function on vertex set. Let $\Gamma(i)=\{j:(i,j)\in E\}$ be the set of vertices adjacent to $i$ in $G$, i.e. the neighborhood of vertex $i$. For each vertex subset $S\subseteq V$, define $w(S)=\sum_{i\in S}w_i$ and $\Gamma(S)=\cup_{i\in S}\Gamma(i)$. Note that it is possible $S\cap \Gamma(S)\neq \emptyset$ and if $S\cap \Gamma(S)=\emptyset$, then $S$ must be independent. For each $S$, define $\alpha(S)=\frac{w(\Gamma(S))}{w(S)}$, referred to as the inclusive expansion ratio of $S$, or the $\alpha$-ratio of $S$ for short. It is not hard to observe that the neighborhood $\Gamma(V)$ is still $V$ and its $\alpha$-ratio is $\alpha(V)=1$.

\begin{definition}[\textbf{Bottleneck and Maximal Bottleneck}]\label{MB}
  A vertex subset $B\subseteq V$ is called a \emph{bottleneck} of $G$ if $\alpha(B)=\min_{S\subseteq V}\alpha(S)$. If bottleneck $B$ is a \emph{maximal bottleneck}, then for any subset $\widetilde{B}$ with $B\subset \widetilde{B}\subseteq V$, it must be $\alpha(\widetilde{B})>\alpha(B)$. We name $(B,\Gamma(B))$ as the maximal bottleneck pair of $G$.
\end{definition}

From Definition \ref{MB}, we can understand the maximal bottleneck as the bottleneck whose size is maximal. Wu and Zhang~\cite{WZ} showed that the maximal bottleneck of any graph is unique and proposed the following bottleneck decomposition with the help of the uniqueness.

\begin{definition}[\textbf{Bottleneck Decomposition}]\label{BD}
  Given an undirected and connected graph $G=(V,E;w)$. Start with $V_1=V$, $G_1=G$ and $i=1$. Find the maximal bottleneck $B_i$ of $G_i$ and let $G_{i+1}$ be the induced subgraph on the vertex set $V_{i+1}=V_i- \left(B_i\cup C_i\right)$, where $C_i=\Gamma(B_i)\cap V_i$, the neighbor set of $B_i$ in the subgraph $G_i$. Repeat if $G_{i+1}\neq \emptyset$ and set $k=i$ if $G_{i+1}=\emptyset$. Then we call $\mathcal{B}=\{(B_1,C_1),\cdots,(B_k,C_k)\}$ the bottleneck decomposition of $G$, $\alpha_i=\frac{w(C_i)}{w(B_i)}$ the $i$-th $\alpha$-ratio and $(\alpha_i)_{i=1}^{k}$ the $\alpha$-ratio vector.
\end{definition}

We propose an example in the following to show the bottleneck decomposition of a graph. Consider the graph of Fig. \ref{example1} which has 6 vertices. The numbers in each circle represents the weight of each vertex. At the first step, $G_1=G$, $V_1=V$ and the maximal bottleneck pair of $G_1$ is
$(B_1,C_1)=(\{v_1,v_2\},\{v_3,v_4\})$ with $\alpha_1=\frac12$. After removing $B_1\cup C_1$, $V_2=\{v_5,v_6\}$, $G_2=G[V_2]$ and the maximal bottleneck pair of $G_2$ is $(B_2,C_2)=(\{v_5,v_6\},\{v_5,v_6\})$ with $\alpha_2=1$. Therefore the bottleneck decomposition is $\mathcal{B}=\{(\{v_1,v_2\},\{v_3,v_4\}),(\{v_5,v_6\},\{v_5,v_6\})\}$.

\begin{figure}[ht]
\centering
\includegraphics[height=2.8cm]{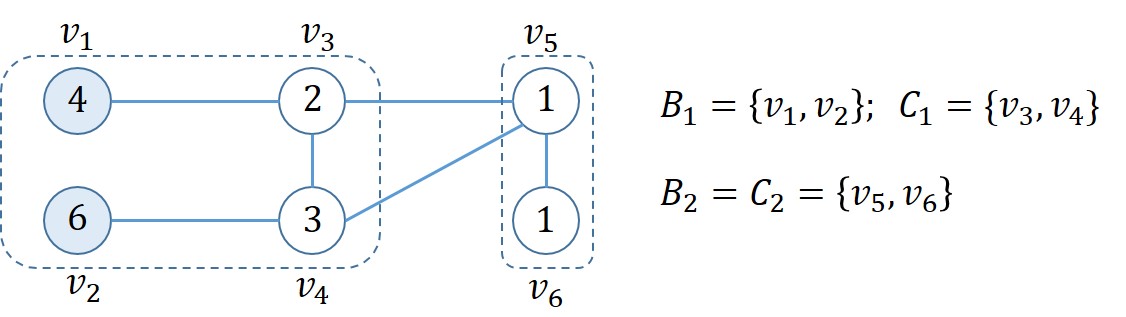}
\caption{\small The figure illustrates bottleneck decomposition of a network.}
\label{example1}
\end{figure}

The bottleneck decomposition has some combinatorial properties which are very crucial to the study on the truthfulness of BD Mechanism in \cite{CDPY} and \cite{CDQY} and to the discussion of the fairness of resource allocation. 
Although Wu and Zhang mentioned these properties in \cite{WZ}, their proofs are not included.
%We explain these properties here and present the proofs in full paper.
We explain these properties here and present the proofs in the Appendix.

\begin{proposition}\label{prop2}
  Given an undirected and connected graph $G=(V,E;w)$, the bottleneck decomposition $\mathcal{B}$ of $G$ satisfies\\
  \noindent (1) $0< \alpha_1<\alpha_2<\cdots<\alpha_k\leq 1$;\\
  \noindent (2) if $\alpha_i=1$, then $i=k$ and $B_k=C_k$; otherwise $B_i$ is independent and $B_i\cap C_i=\emptyset$.
\end{proposition}

\section{Computation of Bottleneck Decomposition}\label{sec3}

In this section, we propose a polynomial time algorithm to compute the bottleneck decomposition of any given network.
Without loss of generality, we assume that all weights of vertices are positive integers and are bounded by $U>0$. From Definition \ref{BD}, it is not hard to see the key to the bottleneck decomposition is the computation of the maximal bottleneck of each subgraph. But how to find the maximal bottleneck among all of $2^{|V|}$ subsets efficiently is a big challenge. To figure it out, our algorithm comprises two phases on each subgraph. In the first phase, we shall compute the minimal $\alpha$-ratio $\alpha^*$ of the current subgraph and find the maximal bottleneck with the minimal $\alpha$-ratio in the second phase.
In the subsequent two subsections, we shall introduce the algorithms in each phase detailedly and analyze them respectively.

\subsection{Evaluating the Minimal $\alpha$-ratio $\alpha^*$}
In this phase to evaluate the minimal $\alpha$-ratio, the main idea of our algorithm is to find the $\alpha$-ratio by binary search approach iteratively. To reach the minimal ratio, we construct a corresponding network with a parameter $\alpha$ and adjust $\alpha$ by applying the maximum flow algorithm until a certain condition is satisfied. Before proceeding the algorithm to compute the minimal $\alpha$-ratio, some definitions and lemmas are necessary.

Given a graph $G=(V,E;w)$ and a parameter $\alpha$, a network $N(G,\alpha)$ (as shown in Fig. \ref{network}-(b)) based on $G$ and $\alpha$ is constructed as:\\
\noindent $\bullet$  $V_N=\{s,t\}\cup V\cup \widetilde{V}$, where $s$ is the source, $t$ is the sink and $\widetilde{V}$ is the copy of $V$;\\
\noindent $\bullet$ the directed edge set $E_N$ comprises:\\
  \qquad - a directed edge $(s,v)\in E_N$ from source $s$ to $v$ with capacity of $\alpha w_v$, $\forall v\in V$;\\
  - a directed edge $(\widetilde{v},t)\in E_N$ from $\widetilde{v}$ to sink $t$ with capacity of $w_v$, $\forall \widetilde{v}\in \widetilde{V}$;\\
  - two directed edges $(u,\widetilde{v})\in E_N$ and $(v,\widetilde{u})\in E_N$ with capacity of $\infty$, $\forall (u,v)\in E$.

\begin{figure}[ht]
\centering
\includegraphics[height=3.5cm]{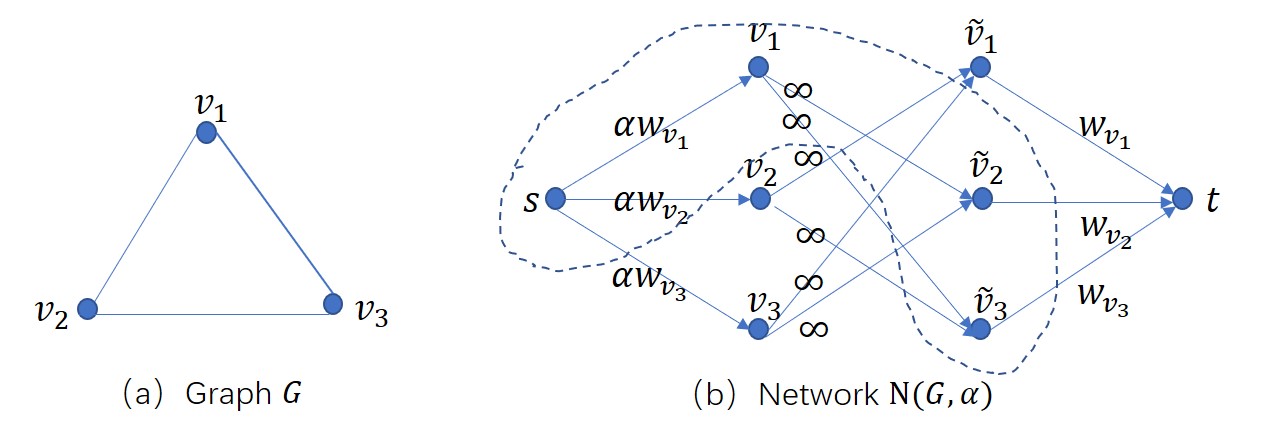}
\caption{\small The illustration of network $N(G,\alpha)$.}\label{network}
\end{figure}

In $N(G,\alpha)$, let $\widetilde{B}=\{\widetilde{v}|v\in B\}$ and $\Gamma(\widetilde{B})=\{\widetilde{u}|u\in \Gamma(B)\}$, where the latter is called the neighborhood of $\widetilde{B}$.

%\begin{definition}[\textbf{$s$-$t$ Cut, Cut Capacity and the Minimum Cut}]
% Given a network $N$. An $s$-$t$ cut $(S,T)$ is a partition of $V_N$ such that $s\in S$ and $t\in T$. The capacity of an $s$-$t$ cut is
% \begin{eqnarray*}
%   cap(S,T)=\sum_{(u,\widetilde{v})\in E_N, u\in S,~ \widetilde{v}\in T}c_{u\widetilde{v}},
% \end{eqnarray*}
% where $c_{u\widetilde{v}}$ is the capacity of edge $(u,\widetilde{v})$. An $s$-$t$ cut is called the minimum cut, if its capacity is minimal among all $s$-$t$ cuts.
%\end{definition}

\begin{lemma}\label{lemma1}
  For any $s$-$t$ cut $(S,T)$ in network $N(G,\alpha)$, the capacity of cut $(S,T)$ is finite if and only if $S$ has the form as $\{s\}\cup B\cup \Gamma(\widetilde{B})$ (as shown in Fig. \ref{network}-(b)) for any subset $B\subseteq V$ and its capacity is $\alpha(w(V)-w(B))+w(\Gamma(B))$.
\end{lemma}

\begin{proof}
  First, if $S=\{s\}\cup B\cup \Gamma(\widetilde{B})$, then there are two kinds of directed edges in cut $(S,T)$:\\
  \noindent $\bullet $ edge $(s,u)$ for any $u\not\in B$ and the total capacity of these edges is $\alpha(w(V)-w(B))$;\\
  \noindent $\bullet$ edge $(\widetilde{u},t)$ for any $\widetilde{u}\in \Gamma(\widetilde{B})$ and the total capacity of these edges is $w(\Gamma(B))$.\\
 So the capacity of cut $(S,T)$ with $S=\{s\}\cup B\cup \Gamma(\widetilde{B})$ is finite and its capacity is $\alpha(w(V)-w(B))+w(\Gamma(B))$.
 Conversely, if $S\neq\{s\}\cup B\cup \Gamma(\widetilde{B})$, then cut $(S,T)$ must contain at least one edge with infinite capacity in the form as $(v,\widetilde{u})$. At this time the capacity of cut $(S,T)$ is infinite. It completes this claim.
 $\square$
\end{proof}

 Lemma \ref{lemma1} tells that if a cut $(S,T)$ in $N(G,\alpha)$ has a finite capacity, set $S$ must has the form as $S=\{s\}\cup B\cup \Gamma(\widetilde{B})$.
 It is not hard to see that such a finite cut corresponds to a subset $B\subseteq V$. Thus we name $B$ as the \emph{corresponding set} of cut $(S,T)$.
 Let $cap(S,T)$ be the capacity of cut $(S,T)$ and $cap(G,\alpha)$ denote the minimum capacity of network $N(G,\alpha)$, that is $cap(G,\alpha)=\min_{(S,T)}cap(S,T)$. To compute the minimal $\alpha$-ratio, we are more interested in the relationship between the current parameter $\alpha$ and the minimal ratio $\alpha^*$.

 \begin{lemma}\label{lemma2}
   Given a graph $G$ and a parameter $\alpha$.
 Let $\alpha^*$ be the minimal $\alpha$-ratio in $G$ and $cap(G,\alpha)$ be the minimum cut capacity of network $N(G,\alpha)$. Then\\
\noindent (1) $cap(G,\alpha)<\alpha w(V)$, if and only if $\alpha>\alpha^*$;\\
\noindent (2) $cap(G,\alpha)=\alpha w(V)$, if and only if $\alpha=\alpha^*$;\\
\noindent (3) $cap(G,\alpha)>\alpha w(V)$, if and only if $\alpha<\alpha^*$.
 \end{lemma}
The proof can be found in the Appendix.
From Lemma \ref{lemma2}-(2), the parameter $\alpha$ is equal to $\alpha^*$, if and only if the corresponding set $B$ of the minimum cut $(S,T)$ in $N(G,\alpha)$ has its $\alpha$-ratio equal to $\alpha$, that is $\frac{w(\Gamma(B))}{w(B)}=\alpha=\alpha^*$. Thus we have the following corollary.

\begin{corollary}\label{coro1}
  For the network $N(G,\alpha)$, if the minimum cut $(S,T)$'s corresponding set $B$ has its $\alpha$-ratio equal to $\alpha$, i.e., $\frac{w(\Gamma(B))}{w(B)}=\alpha$, then $\alpha=\alpha^*$.
\end{corollary}

Based on Lemma~\ref{lemma1} and~\ref{lemma2} and Corollary \ref{coro1}, following Algorithm A is derived to compute the minimal $\alpha$-ratio for any given graph $G=(V,E;w)$.

\begin{table}[htbp]
\caption{Algorithm A: Compute the minimal $\alpha$-ratio of $G$}
\label{alg1}
\small
\centering
\begin{tabular}{l}
\hline
\textbf{Input}: Graph $G=(V,E;w)$\\
\textbf{Output}: The minimal $\alpha$-ratio $\alpha^*$ of $G$.\\
\hline
1: \textbf{Set} $a:=0$, $b:=1$ and $M:=w^2(V)$;\\
2: \textbf{Set} $\alpha:=\frac{1}2(a+b)$;\\
3: \textbf{Constrcut} network $N(G,\alpha)$;\\
4: \textbf{Compute} minimum s-t cut $(S,T)$ with its capacity $cap(G,\alpha)$ by Edmonds-Karp \\ 
~~~~algorithm and obtain the corresponding subset $B\subseteq V$;\\
5: \textbf{If} $|\frac{w(\Gamma(B))}{w(B)}-\alpha|<\frac1{M}$\\
    \hspace{0.35in} \textbf{Output} $\alpha^*=\frac{w(\Gamma(B))}{w(B)}$;\\
6: \textbf{Else} \\
7: \hspace{0.2in} \textbf{If} $cap(G,\alpha)>\alpha w(V)$ \\
    \hspace{0.6in} \textbf{Set} $a:=\alpha$, $b:=b$ and turn to line 2;\\
8: \hspace{0.2in} \textbf{If} $cap(G,\alpha)<\alpha w(V)$\\
    \hspace{0.6in} \textbf{Set} $a:=a$, $b:=\alpha$ and turn to line 2.\\
\hline
\end{tabular}
\end{table}

The main idea of Algorithm A is to find $\alpha^*$ by binary search approach. The initial range of $\alpha^*$ is set as $[0,1]$, since $0< \alpha^*\leq 1$ by Proposition \ref{prop2}-(1). And in each iteration, we split the range in half by comparing the minimum capacity $cap(G,\alpha)$ and the value of $\alpha w(V)$. Because all weights of vertices are positive integers, each subset's $\alpha$-ratio is a rational number and the difference of any two different $\alpha$-ratios should be greater than $\frac1{w^2(V)}$. So once $\left|\frac{w(\Gamma(B))}{w(B)}-\alpha\right|<\frac{1}{w^2(V)}$, we can conclude $\frac{w(\Gamma(B))}{w(B)}=\alpha$ and Corollary \ref{coro1} makes us get $\frac{w(\Gamma(B))}{w(B)}=\alpha=\alpha^*$. Thus we set the terminal condition of of Algorithm A as $\left|\frac{w(\Gamma(B))}{w(B)}-\alpha\right|<\frac{1}{w^2(V)}$.

There is only one loop in Algorithm A. 
After each round, the length of the search range $[a,b]$ is cut in half. 
Thus the loop ends within $O(log(M))$ rounds, that is $O(log(w(V)))$. 
In each round, the critical calculation is on line 4 to compute the minimum $s$-$t$ cut of a given network. 
Applying the famous \emph{max-flow min-cut theorem}\cite{PS}, it is equivalent to compute the corresponding maximum flow. 
There are several polynomial-time algorithms to find maximum flow, for instance \emph{Edmonds-Karp algorithm}\cite{EK} with time complexity $O(|V||E|^2)$. 
In Algorithm A, we call the Edmonds-Karp algorithm to compute the minimum capacity and have the following theorem.

\begin{theorem}\label{theo1}
The minimal $\alpha$-ratio $\alpha^*$ of $G$ can be computed in $O(|V||E|^2log(w(V)))$ time.
\end{theorem}

\subsection{Finding the Maximal Bottleneck}

In the previous phase, we compute the minimal $\alpha$-ratio $\alpha^*$ by Algorithm A and the corresponding bottleneck $B'$ also can be obtained. 
But it is possible that the bottleneck $B'$ may not be maximal. 
So in this phase, we continue to find the maximal bottleneck given the minimal $\alpha$-ratio of $\alpha^*$.

Here we introduce another network, denoted by $N(G,\alpha^*,\epsilon)$, for a given parameter $\epsilon>0$.
To obtain network $N(G,\alpha^*,\epsilon)$, we first construct network $N(G,\alpha^*)$, defined in Section 3.1, and then increase the capacity of edge $(s,v)$, for any $v\in V$, from $\alpha^*w_v$ to $\alpha^*w_v+\epsilon$. 
By similar proof for Lemma \ref{lemma1}, we know a cut $(S,T)$ in $N(G,\alpha^*,\epsilon)$ has a finite capacity, if and only if $S$ has the form as $S=\{s\}\cup B\cup \Gamma(\widetilde{B})$, where $B\subseteq V$.
And the corresponding capacity becomes
\begin{eqnarray}
  cap(S,T)&=&\alpha^*(w(V)-w(B))+(n-|B|)\cdot\epsilon+w(\Gamma(B))\nonumber\\
  &=&\alpha^*w(V)-w(B)(\alpha^*-\frac{w(\Gamma(B))}{w(B)})+(n-|B|)\cdot\epsilon\label{eqn4}.
\end{eqnarray}

From (\ref{eqn4}), the capacity $cap(S,T)$ actually depends on its corresponding set $B$ and the parameter $\epsilon$. Thus we can view it as a function of $B$ and $\epsilon$. To simplify our discussion in this section, we use $cap(B,\epsilon)$, different from the notation in the previous subsection, to represent the capacity function of cut $(S,T)$ where $S=\{s\}\cup B\cup \Gamma(\widetilde{B})$.

\begin{lemma}\label{lemma3}
Given a graph $G$, let $B^*$ be the maximal bottleneck of $G$. For any $\epsilon>0$, the maximal bottleneck $B^*$ satisfies $cap(B^*,\epsilon)=\min_{B:\alpha(B)=\alpha^*}cap(B,\epsilon)$ in network $N(G,\alpha^*,\epsilon)$.
\end{lemma}
\begin{proof} For any bottleneck $B$ with $\alpha(B)=\alpha^*$, we know $
cap(B,\varepsilon)
=\alpha^*w(V)+n\cdot\varepsilon-\varepsilon|B|.$
 Thus $cap(B^*,\varepsilon)=\min_{\alpha(B)=\alpha^*}cap(B,\varepsilon)$ since $B^*$ is the maximal bottleneck of $G$.
 $\square$
\end{proof}

In other word, if the minimum cut $\hat{B}$ in $N(G,\alpha^*,\epsilon)$, with proper parameter $\epsilon$, is a bottleneck, then $cap(\hat{B},\epsilon)=\min_{B\subseteq V}cap(B,\epsilon)=\min_{B:\alpha(B)=\alpha^*}cap(B,\epsilon)$, which means $\hat{B}$ is the maximal bottleneck of $G$.
% Therefore Lemma \ref{lemma3} ensures that such a set $\hat{B}$ must be the maximal bottleneck.

\begin{corollary} \label{lemma4}
Given a graph $G$.
  If there is an $\epsilon>0$ such that the corresponding set $\hat{B}$ of the minimum cut in $N(G,\alpha^*,\epsilon)$ is a bottleneck, then $\hat{B}$ is the maximal bottleneck of $G$.
\end{corollary}

Furthermore, once the $\epsilon$ is small enough, the corresponding set of the minimum cut is the maximal bottleneck of $G$, shown in Lemma~\ref{lemma6} whose detailed proof is presented in the Appendix. 
%(We leave the proof in full paper due to space limit.)

\begin{lemma}\label{lemma6}
Given a graph $G$. If $\epsilon\leq \frac1{w^3(V)}$, then the corresponding set $\hat{B}$ of the minimum cut in $N(G,\alpha^*,\epsilon)$ is the maximal bottleneck of $G$.
\end{lemma}
% \begin{proof} Denote $B^*$ to be the maximal bottleneck of $G$. In network $N(G,\alpha^*,\epsilon)$,
% \begin{eqnarray*}
%   \alpha^*w(V)+(\alpha(\hat{B})-\alpha^*)w(\hat{B})+(n-|\hat{B}|)\cdot\epsilon&=&cap(\hat{B},\epsilon)\\
%   &\leq&cap(B^*,\epsilon)=\alpha^*w(V)+(n-|B^*|)\cdot\epsilon.
% \end{eqnarray*}
% The inequality comes from the condition that $cap(\hat{B},\epsilon)$ is the minimum cut capacity in $N(G,\alpha^*,\epsilon)$. So
% $0\leq (\alpha(\hat{B})-\alpha^*)w(\hat{B})\leq \left(|\hat{B}|-|B^*|\right)\cdot\epsilon$
% and
% $$
%   \alpha(\hat{B})-\alpha^*\leq \frac{|\hat{B}|-|B^*|}{w(\hat{B})}\cdot\epsilon
%   \leq\frac{|\hat{B}|-|B^*|}{w(\hat{B})}\cdot\frac{1}{w^3(V)}<\frac{1}{w^2(V)}.
% $$
% As assumed in advance that the weights of all vertices are positive integers, each set's $\alpha$-ratio is a rational number and the difference of any two different $\alpha$-ratios should be great than $\frac1{w(V)^2}$.
% Because $\alpha(\hat{B})-\alpha^*<\frac{1}{w^2(V)}$, we can confirm $\alpha(\hat{B})=\alpha^*$, which means $\hat{B}$ is a bottleneck. In addition, Corollary \ref{lemma4} makes it sure that $\hat{B}$ is the maximal bottleneck of $G$.
% $\square$
% \end{proof}

Based on Lemma~\ref{lemma6}, we propose the following Algorithm B to find the maximal bottleneck of a graph $G$ if its minimal $\alpha$-ration is given beforehand.

\begin{table}[htbp]
\caption{Algorithm B: Find the maximal bottleneck of $G$}
\label{alg2}
\small
\centering
\begin{tabular}{l}
\hline
\textbf{Input}: Graph $G$ and its minimal $\alpha$-ratio $\alpha^*$;\\
\textbf{Output}: The maximal bottleneck $B^*$ of $G$.\\
\hline
1: \textbf{Set} $\epsilon:=\frac1{w^3(V)}$;\\
2: \textbf{Construct} network $N(G,\alpha^*,\epsilon)$;\\
3: \textbf{Compute} the minimum cut capacity $cap(\hat{B},\epsilon)$ by Edmonds-Karp algorithm\\
 ~~~~and obtain the corresponding set $\hat{B}\subseteq V$;\\
4: \textbf{Output} $B^*=\hat{B}$;\\
\hline
\end{tabular}
\end{table}

Lemma \ref{lemma6} guarantees Algorithm B outputs the maximal bottleneck $B^*$ of $G$ correctly. 
The main body of Algorithm B is to compute the minimum cut capacity of network $N(G,\alpha^*,\epsilon)$ which can be realized by \emph{Edmonds-Karp algorithm}. 
So the time complexity of Algorithm B is $O(|V||E|^2)$.

\begin{theorem}\label{theo1.1}
  Algorithm B ouputs the maximal bottleneck of $G$ in $O(|V||E|^2)$ time.
\end{theorem}

Applying Algorithm A and B, the maximal bottleneck of any given graph can be computed. Thus by Definition \ref{BD}, we can get the bottleneck decomposition of $G$ by iteratively calling Algorithm A and B on each subgraph. The main result of this paper is:
\begin{theorem}\label{theo2} Given a graph $G=(V,E;w)$,
the bottleneck decomposition of can be computed in $O(n^6log(nU))$ time, where $n=|V|$ and $U=\max_{v\in V}w_v$.
\end{theorem}

\begin{proof} 
To compute the bottleneck decomposition, Algorithm A and B are run repeatedly. 
In each round we obtain the maximal bottleneck and its neighborhood, then delete them and go to the next round.
So the time complexity of each round is $O(|V||E|^2log(w(V)))$. 
At the end of each round at least one vertex is removed. 
Thus the bottleneck decomposition contains at most $O(|V|)$ loops, which means the total time complexity is $O(|V|^2|E|^2log(w(V)))$. 
Since $|E|$ is at most $O(|V|^2)$ and the weight of each vertex is bounded by $U$, the time complexity can be written as $O(|V|^6log(|V|U))=O(n^6log(nU))$, if $|V|=n$.
$\square$
\end{proof}

\section{Bottleneck Decomposition, Market Equilibrium and Fair Allocation}\label{sec4}
To derive an allocation efficiently, Wu and Zhang \cite{WZ} modeled the resource exchange system as a pure exchange economy, and obtain the equilibrium allocation by computing a market equilibrium. 
% In this section, we shall introduce the reduction in \cite{WZ} from the bottleneck decomposition to a market equilibrium. 
In this section, we shall present some properties of it, and further prove the allocation of such a market equilibrium not only has the property of proportional response, but also is lex-optimal.

\begin{definition}[\textbf{Market Equilibrium}]
Let $p_i$ be the price of agent $i$'s whole resource, $1\leq i\leq n$. The price vector $p=(p_1,p_2,\cdots,p_n)$, with the allocation $X=(x_{ij})$ is called a \emph{market equilibrium} if for each agent  $i\in V$ the following holds:\\
 \noindent 1. $\sum_{j\in \Gamma(i)}x_{ij}=1$ (market clearance);\\
 \noindent 2. $\sum_{j\in \Gamma(i)}x_{ji}p_j\leq p_i$ (budget constraint);\\
 \noindent 3. $X=(x_{ij})$ maximizes $\sum_{j\in \Gamma(i)}x_{ji}w_j$, s.t. $\sum_{j\in \Gamma(i)}x_{ji}p_j\leq p_i$ and $x_{ij}\geq 0$ for each vertex $i$ (individual optimality).
\end{definition}

\noindent \textbf{Construction of a Market Equilibrium from Bottleneck Decomposition}

Given the bottleneck decomposition $B=\{(B_1,C_1),\cdots,(B_k,C_k))\}$, an allocation can be computed by distinguishing three cases~\cite{WZ}. For convenience, such an allocation mechanism is named as \emph{BD Mechanism} by Cheng \emph{et al}~\cite{CDPY,CDQY}. Fig.~\ref{bd} well illustrates it.

\noindent\textbf{BD Mechanism:}
\begin{itemize}
\item For $\alpha_i<1$ (i.e., $B_i\cap C_i=\emptyset$), consider the bipartite graph $\hat{G}_i=(B_i,C_i; E_i)$ where $E_i=(B_i\times C_i)\cap E$. Construct a network by adding source $s$, sink $t$ and directed edge $(s,u)$ with capacity $w_u$ for any $u\in B_i$, directed edge $(v,t)$ with capacity $\frac{w_v}{\alpha_i}$ for any $v\in C_i$ and directed edge $(u,v)$ with capacity $\infty$ for any $(u,v)\in E_i$. By the  max-flow min-cut theorem, there exists flow $f_{uv}\geq 0$ for $u\in B_i$ and $v\in C_i$ such that $\sum_{v\in \Gamma(u)\cap C_i}f_{uv}=w_u$ and $\sum_{u\in \Gamma(v)\cap B_i}f_{uv}=\frac{w_v}{\alpha_i}$. Let the allocation be $x_{uv}=\frac{f_{uv}}{w_u}$ and $x_{vu}=\frac{\alpha_if_{uv}}{w_v}$ which means that $\sum_{v\in \Gamma(u)\cap C_i}x_{uv}=1$ and $\sum_{u\in \Gamma(v)\cap B_i}x_{vu}=\sum_{u\in \Gamma(v)\cap B_i}\frac{\alpha_i\cdot f_{vu}}{w_v}=1$.
\item For $\alpha_k=1$ (i.e., $B_k=C_k=V_k$), construct a bipartite graph $\hat{G}=(B_k,B'_k;E'_k)$ such that $B'_k$ is a copy of $B_k$, there is an edge $(u,v')\in E'_k$ if and only if  $(u,v)\in E[B_k]$.
    Construct a network by the above method and
    by Hall's theorem, for any edge $(u,v')\in E'_k$, there exists flow $f_{uv'}$ such that $\sum_{v'\in \Gamma(u)\cap B'_k}f_{uv'}=w_u$. Let the allocation be $x_{uv}=\frac{f_{uv'}}{w_u}$.
\item For any other edge, $(u,v)\not\in B_i\times C_i$, $i=1,2,\cdots,k$, define $x_{uv}=0$.
\end{itemize}

\begin{figure}[ht]
\centering
\includegraphics[height=4.5cm]{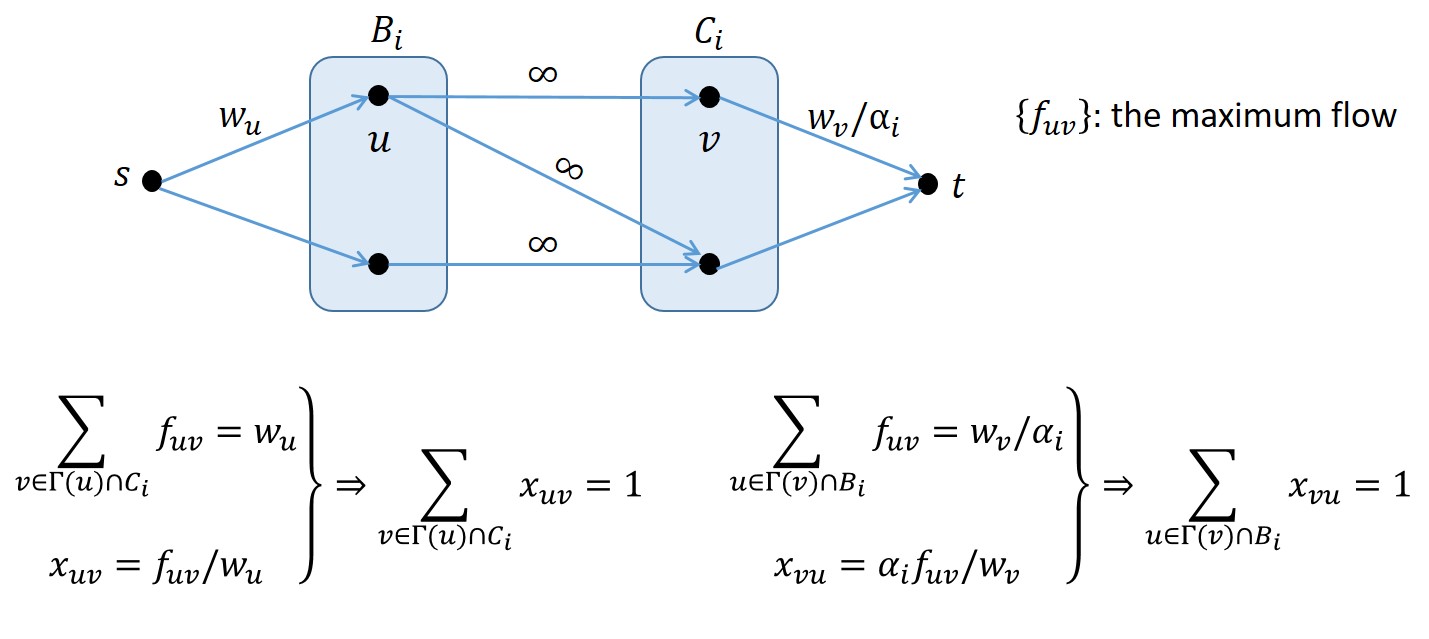}
\caption{\small The illustration of BD Mechanism.}\label{bd}
\end{figure}

Wu and Zhang stated Proposition~\ref{prop4}, saying that once the prices of resource are set properly, such a price vector and the allocation from BD Mechanism make up a market equilibrium, for which the proof is omitted in \cite{WZ}. 
%To make readers understand clearly, we also propose the detailed proof in full paper.
To make readers understand clearly, we also propose the detailed proof in the Appendix.
\begin{proposition}[\cite{WZ}]\label{prop4}
  Given $B=\{(B_1,C_1),\cdots,(B_k,C_k))\}$. If the price to each vertex is set as: for $u\in B_i$, let $p_u=\alpha_iw_u$; and for $u\in C_i$, let $p_u=w_u$, then $(p,X)$ is a market equilibrium, where $X$ is the allocation from BD Mechanism. Furthermore, each agent $u$'s utility is $U_u=w_u\cdot \alpha_i$, if $u\in B_i$; otherwise $U_u=w_u/ \alpha_i$.
\end{proposition}

Motivated by P2P systems, such as BitTorrent, the concept of proportional response for the consideration of fairness among all participating agents is put forward to encourage the agents to join in the P2P system.
\begin{definition}[\textbf{Proportional Response}]
For each agent $i$, the allocation $(x_{uv}: v\in \Gamma(u))$ of his resource $w_u$ is proportional to what he receives from his neighbors $(w_v\cdot x_{vu}: v\in \Gamma(u))$, i.e.,
$x_{uv}=\frac{x_{vu}w_v}{\sum_{k\in \Gamma(u)}x_{ku}w_k}=\frac{x_{vu}w_v}{U_u}.$
\end{definition}

\begin{proposition}\label{prop7}
  The allocation $X$ from BD Mechanism satisfies the property of proportional response.
\end{proposition}
\begin{proof}
  For the allocation from BD Mechanism, if $u\in B_i$ and $v\in C_i$, then $ U_u=w_u\cdot \alpha_i$, $U_v=w_v/\alpha_i$ and $f_{uv}=x_{uv}\cdot w_u= \frac{x_{vu}\cdot w_v}{\alpha_i}.$ So $ x_{uv}=\frac{f_{uv}}{w_u}=\frac{x_{vu}w_v}{\alpha_iw_u}=\frac{x_{vu}w_v}{U_u}$ and $
    x_{vu}=\frac{\alpha_if_{uv}}{w_v}=\frac{x_{uv}w_u}{w_v/\alpha_i}=\frac{x_{uv}w_u}{U_v}$.
$\square$
\end{proof}

Clearly, the allocation $X$ from BD Mechanism can be computed from the maximum flow in each bottleneck pair $(B_i,C_i)$, by Edmonds-Karp Algorithm. So the total time complexity of BD Mechanism is $O(n^5)$. Combining Theorem \ref{theo2} and Proposition \ref{prop7}, we have
\begin{theorem}\label{theo3}
  In the resource sharing system, an allocation with the property of proportional response can be computed in $O(n^6\log(nU))$.
\end{theorem}

Recently, Georgiadis \emph{et al.}~\cite{GIT} discuss the fairness from the compensatory point of view. They characterized the exchange performance of an allocation by the concept of \emph{exchange ratio} vector, in which the coordinate is the exchange ratio $\beta_u(X)=U_u(X)/w_u$ of each agent. In~\cite{GIT}, an allocation is said to be fair, if its exchange ratio vector is lex-optimal, and its properties are introduced in the following. Here some notations shall be introduced in advance. Given an allocation $X$, for a set $S\subseteq V$, $N(S)=\{v\in V:x_{uv}>0, \exists u\in S\}$ denotes the set of agents who receive resource from agents in $S$. For the exchange ratio vector $\beta(X)$, the different values (level) of coordinates are denoted by $l_1<l_2<\cdots<l_M$, $M\leq n$. Let $L_i(X)=\{v\in V: \beta_v(X)=l_i\}$ be the set in which each agent's exchange ratio is equal to $l_i$. Georgiadis et al.~\cite{GIT} proposed the following characterization of a lex-optimal allocation.

\begin{proposition}[\cite{GIT}]\label{prop5}
  1. An allocation $X$ with $M\geq 2$ is lex-optimal, if and only if\\
  \noindent (1) $L_i$ is an independent set in $G$, $i=1,2,\cdots,\lfloor\frac M2\rfloor$;\\
  \noindent (2) $L_{M-i+1}=N(L_i)$, $i=1,2,\cdots,\lfloor\frac M2\rfloor$;\\
  \noindent (3) $l_i\cdot l_{M-i+1}=1$, $i=1,2,\cdots,\lfloor\frac M2\rfloor$;\\
  \noindent (4) $\sum_{u\in L_i}U_u=\sum_{u\in L_{M-i+1}}w_u$, $i=1,2,\cdots,\lfloor\frac M2\rfloor$.\\
  \noindent 2. An allocation $X$ with $M=1$ is lex-optimal if and only if $l_1=1$. %Also if an allocation $X$ has $M=1$ and $l_1=1$, then $X$ is lex-optimal.
\end{proposition}

Based on above proposition, we can continue to conclude that the allocation $X$ from BD Mechanism is also lex-optimal.

% The second claim in Proposition \ref{prop5} shows the sufficient and necessary condition of the most ideal state that each agent's exchange ratio is equal to 1.

\begin{theorem}\label{prop6}
  The allocation $X$ from BD Mechanism is lex-optimal.
\end{theorem}
\begin{proof}
  Given a bottleneck decomposition $\mathcal{B}=\{(B_1,C_1),\cdots,(B_k,C_k)\}$.  
  By Proposition \ref{prop4}, $U_u=w_u\cdot\alpha_i$ if $u\in B_i$ and $U_u=w_u/\alpha_i$ if $u\in C_i$. 
  Thus each agent's exchange ratio can be written as: $\beta_u=\alpha_i$ if $u\in B_i$ and $U_u=1/\alpha_i$ if $u\in C_i$.
  If $k=1$ and $B_1=C_1=V$, then all agents have the same exchange ratio $\beta_u(X)=1=l_1$ with $M=1$ and the second claim in Proposition \ref{prop5} is satisfied for this case. 
  If $k=1$ and $\alpha_k<1$ or $k>1$, then we know $\alpha_1<\alpha_2<\cdots<\alpha_k\leq 1$ by Proposition \ref{prop2}. 
  The relationship between $\alpha$-ratio and exchange ratio makes the different values of $\beta_u$ be ordered as: $\alpha_1<\cdots<\alpha_k\leq 1/\alpha_k<\cdots<1/\alpha_1$, where $\alpha_k=1/\alpha_k$ if and only if $\alpha_k=1$ and $B_k=C_k$. 
  So the number of different values $M=2k$ if $\alpha_k<1$ and $M=2k-1$ if $\alpha_k=1$. 
  By the definitions of $l_i$ and $L_i$, we have $L_i=B_i$ with $l_i=\alpha_i<1$ and $L_{M-i+1}=C_i$ with $l_{M-i+1}=1/\alpha_i>1$ and $L_i=B_i$ is independent by Proposition \ref{prop2}, $i=1,\cdots,\lfloor M/2\rfloor$. 
  In addition, since all resource exchange only happens between $B_i$ and $C_i$ by BD Mechanism, $L_{M-i+1}=N(L_i)$ and $\sum_{u\in L_i}U_u=\sum_{v\in L_{M-i+1}}w_{v}$, $i=1,\cdots, \lfloor M/2\rfloor$. Until now all statements of the first claim are satisfied for this case. 
  It means the allocation $X$ from BD Mechanism is lex-optimal.
  $\square$
\end{proof}

\section{Conclusion}\label{sec5}
This paper discusses the issue of the computation of a fair allocation in the resource sharing system through a combinatorial bottleneck decomposition. 
We design an algorithm to solve the bottleneck decomposition for any graph $G(V,E;w_v)$ in $O(n^6\log(nU))$ time, where $n=|V|$ and $U=\max_{v\in V} w_v$. 
Our work also completes the computation of a market equilibrium in the resource exchange system for the consideration of economic efficiency in~\cite{WZ}. 
% Another one is to make the analysis of agents' incentives be complete from the perspective of mechanism design \cite{CDPY,CDQY}. 
Furthermore, we show the equilibrium allocation from the bottleneck decomposition not only is proportional response, but also is lex-optimal, which establishes a connection between two concepts of fairness in~\cite{GIT} and~\cite{WZ}. 
% Compared with the algorithm for a lex-optimal allocation in~\cite{GIT} by transforming to a linear programming, our algorithm is totally different based on a combinatorial structure: bottleneck decomposition.
Involving two different definitions of fairness for resource allocation, we hope to explore other proper concepts of fairness and to design efficient algorithms to find such fair allocations in the future.

%
% the environments 'definition', 'lemma', 'proposition', 'corollary',
% 'remark', and 'example' are defined in the LLNCS documentclass as well.
\bibliographystyle{splncs04}
\bibliography{reference}




% LNCS Homepage, \url{http://www.springer.com/lncs}. Last accessed 4
% \end{thebibliography}

% \newpage

 \section{Appendix}
 \addtocounter{lemma}{-3}
 \begin{lemma}
   Given a graph $G$ and a parameter $\alpha$.
  Let $\alpha^*$ be the minimal $\alpha$-ratio in $G$ and $cap(G,\alpha)$ be the minimum cut capacity of network $N(G,\alpha)$. Then\\
 \noindent (1) $cap(G,\alpha)<\alpha w(V)$, if and only if $\alpha>\alpha^*$;\\
 \noindent (2) $cap(G,\alpha)=\alpha w(V)$, if and only if $\alpha=\alpha^*$;\\
 \noindent (3) $cap(G,\alpha)>\alpha w(V)$, if and only if $\alpha<\alpha^*$.
  \end{lemma}
 \begin{proof}
   Let $(S,T)$ be the minimum cut of the network $N(G,\alpha)$ and $B'$ be a bottleneck in $G$ satisfying $\frac{w(\Gamma(B'))}{w(B')}=\alpha^*$.
   Therefore cut $(S,T)$ must have a finite capacity and there exists a subset $B$ such that $S=\{s\}\cup B\cup\Gamma(\widetilde{B})$ by Lemma \ref{lemma1}. It is possible that $B\neq B'$. So $\frac{w(\Gamma(B))}{w(B)}\geq \alpha^*$. The minimum capacity of $(S,T)$ is
 $
     cap(G,\alpha)=\alpha(w(V)-w(B))+w(\Gamma(B))=\alpha w(V)+w(B)(\frac{w(\Gamma(B))}{w(B)}-\alpha).
 $

   (1) If $cap(G,\alpha)<\alpha w(V)$, it is easy to deduce that $\alpha>\frac{w(\Gamma(B))}{w(B)}\geq \alpha^*$. Conversely, if $\alpha>\alpha^*$, we can construct another cut $(S',T')$ where $S'=\{s\}\cup B'\cup \Gamma(\widetilde{B}')$ corresponding to the bottleneck $B'$. Then $cap(S',T')=\alpha(w(V)-w(B'))+w(\Gamma(B'))=\alpha w(V) + (\alpha^*-\alpha) w(B')<\alpha w(V).$
 Therefore, $cap(G,\alpha)\leq cap(S',T')<\alpha w(V)$.

   (2) If $cap(G,\alpha)=\alpha w(V)$, then $\alpha=\frac{w(\Gamma(B))}{w(B)}\geq \alpha^*$. It is not hard to see if $\alpha>\alpha^*$, then $cap(G,\alpha)<\alpha w(V)$ by Claim (1), contradicting to the condition of $cap(G,\alpha)=\alpha w(V)$. So $\alpha=\alpha^*$.  On the other hand,
   if $\alpha=\alpha^*$, then the fact $\frac{w(\Gamma(B))}{w(B)}\geq \alpha^*$ makes
     $cap(G,\alpha)=\alpha w(V)+w(B)(\frac{w(\Gamma(B))}{w(B)}-\alpha^*)\geq\alpha w(V)$. Now let us construct cut $(S',T')$ where $S'=\{s\}\cup B'\cup \Gamma(\widetilde{B}')$. Obviously $cap(G,\alpha)\leq cap(S',T')=\alpha w(V)$. Combining above two aspects, we have $cap(G,\alpha)=\alpha w(V)$

 (3) For the case $cap(G,\alpha)>\alpha w(V)$, to prove $\alpha<\alpha^*$, we suppose to the contrary that $\alpha\geq \alpha^*$. Then Claim (1) and (2) promise that $cap(G,\alpha)\leq \alpha w(V)$. It's a contradiction. Thus $\alpha<\alpha^*$. On the other hand, if $\alpha<\alpha^*$, we can easily infer that
 $cap(G,\alpha)=\alpha w(V)+w(B)(\frac{w(\Gamma(B))}{w(B)}-\alpha)> \alpha w(V),$
 where the inequality is right since $\frac{w(\Gamma(B))}{w(B)}\geq \alpha^*>\alpha$.
 $\square$
 \end{proof}

 \addtocounter{lemma}{+1}
 \begin{lemma}
 Given a graph $G$. If $\epsilon\leq \frac1{w^3(V)}$, then the corresponding set $\hat{B}$ of the minimum cut in $N(G,\alpha^*,\epsilon)$ is the maximal bottleneck of $G$.
 \end{lemma}
 \begin{proof}
 Denote $B^*$ to be the maximal bottleneck of $G$. In network $N(G,\alpha^*,\epsilon)$,
 \begin{eqnarray*}
   \alpha^*w(V)+(\alpha(\hat{B})-\alpha^*)w(\hat{B})+(n-|\hat{B}|)\cdot\epsilon&=&cap(\hat{B},\epsilon)\\
   &\leq&cap(B^*,\epsilon)\\
   &=&\alpha^*w(V)+(n-|B^*|)\cdot\epsilon.
 \end{eqnarray*}
 The inequality comes from the condition that $cap(\hat{B},\epsilon)$ is the minimum cut capacity in $N(G,\alpha^*,\epsilon)$. So
 $0\leq (\alpha(\hat{B})-\alpha^*)w(\hat{B})\leq \left(|\hat{B}|-|B^*|\right)\cdot\epsilon$
 and
 $$
   \alpha(\hat{B})-\alpha^*\leq \frac{|\hat{B}|-|B^*|}{w(\hat{B})}\cdot\epsilon
   \leq\frac{|\hat{B}|-|B^*|}{w(\hat{B})}\cdot\frac{1}{w^3(V)}<\frac{1}{w^2(V)}.
 $$
 As assumed in advance that the weights of all vertices are positive integers, each set's $\alpha$-ratio is a rational number and the difference of any two different $\alpha$-ratios should be great than $\frac1{w(V)^2}$.
 Because $\alpha(\hat{B})-\alpha^*<\frac{1}{w^2(V)}$, we can confirm $\alpha(\hat{B})=\alpha^*$, which means $\hat{B}$ is a bottleneck. In addition, Corollary \ref{lemma4} makes it sure that $\hat{B}$ is the maximal bottleneck of $G$.
 $\square$
 \end{proof}

 \addtocounter{proposition}{-4}
 \begin{proposition}
   Given an undirected and connected graph $G=(V,E;w)$, the bottleneck decomposition $\mathcal{B}$ of $G$ satisfies\\
   \noindent (1) $0< \alpha_1<\alpha_2<\cdots<\alpha_k\leq 1$;\\
   \noindent (2) if $\alpha_i=1$, then $i=k$ and $B_k=C_k$; otherwise $B_i$ is independent and $B_i\cap C_i=\emptyset$.
 \end{proposition}
 \begin{proof}
 For the first claim, we try to show $0<\alpha_i\leq 1$ for each $i=1,2,\cdots,k$ firstly and then prove the monotonic increasing property of sequence $\{\alpha_i\}$.

 To achieve the result of $0<\alpha_i\leq 1$, we shall prove that
  each subgraph $G_i$ does not contain any isolated vertex. Thus based on the claim of no isolated vertex, any vertex subset's neighborhood in $G_i$ is not empty and $\alpha_i$ is larger than $0$. Meanwhile, we have a special vertex set $V_i$, whose neighborhood in $G_i$ is still $V_i$, because there is no isolated vertex in $G_i$. It implies $V_i$'s $\alpha$-ratio in $G_i$ is 1. Therefore $\alpha_i$ can not larger than 1 and it must be $\alpha_i\leq 1$.

  To show there is no isolated vertices in any $G_i$, $i>1$, we suppose to the contrary that $G_{i}$ is the subgraph with the smallest index in $\mathcal{B}$ containing an isolated vertex $x$. Then $x$ must has neighbors in $B_{i-1}\cup C_{i-1}$. Otherwise, $x$ also is isolated in $G_{i-1}$, which contradicts to our assumption of the smallest index on $G_i$. We claim $x$ does not have neighbors in $B_{i-1}$. If not, at least one neighbor of $x$ is in $B_{i-1}$. Symmetrically, $x$ also is a neighbor of $B_{i-1}$ and it shall be in $C_{i-1}=\Gamma(B_{i-1})\cap V_{i-1}$. So the bottleneck decomposition tells us $x$ should be removed before the $i$-th decomposition. It's a contradiction to the condition that $x\in G_i$. Since $x$'s neighbors in $G_{i-1}$ are all contained in $C_{i-1}$, let us consider another set $B_{i-1}\cup \{x\}$ whose neighborhood in $V_{i-1}$ still is $C_{i-1}$. Obviously, its $\alpha$-ratio is $\frac{w(C_{i-1})}{w(B_{i-1})+w_x}<\alpha_{i-1}$, contradicting the fact $(B_{i-1},C_{i-1})$ is the maximal bottleneck pair in $G_{i-1}$.

 To prove the monotonic increasing property of sequence $\{\alpha_i\}$, we suppose to the contrary that there exists an index $i\in \{1,2,\cdots,k\}$, such that $\alpha_i\geq \alpha_{i+1}$, and focus on two pairs $(B_i,C_i)$ and $(B_{i+1},C_{i+1})$. Recall Definition \ref{BD} of bottleneck decomposition, $G_{i+1}=G_i-(B_i\cup C_i)$. Thus $B_{i+1}\cap B_{i}=\emptyset$, $C_{i+1}\cap C_{i}=\emptyset$ and $B_i\cup B_{i+1}\subseteq V_i$, $\Gamma(B_i\cup B_{i+1})\cap V_{i}=C_i\cup C_{i+1}$. And the pair $(B_{i}\cup B_{i+1},\Gamma(B_i\cup B_{i+1})\cap V_{i})$ in $G_i$ has its $\alpha$-ratio
   \begin{eqnarray}
     \frac{w(\Gamma(B_i\cup B_{i+1})\cap V_{i})}{w(B_{i}\cup B_{i+1})}=\frac{w(C_i\cup C_{i+1})}{w(B_i\cup B_{i+1})}=\frac{w(C_i)+w(C_{i+1})}{w(B_i)+w(B_{i+1})}\leq \alpha_{i},\label{eqn2}
   \end{eqnarray}
   where the last inequality is from the assumption that $\frac{w(C_i)}{w(B_i)}=\alpha_i\geq \alpha_{i+1}=\frac{w(C_{i+1})}{w(B_{i+1})}$.
   On the other hand, we know $(B_i,C_i)$ is the maximal bottleneck of $G_i$ with $\alpha$-ratio $\alpha_i$. The definition of maximal bottleneck makes $\frac{w(\Gamma(B_i\cup B_{i+1})\cap V_{i})}{w(B_{i}\cup B_{i+1})}> \alpha_i$ since $B_i\subset B_i\cup B_{i+1}$. It contradicts to (\ref{eqn2}). So there does not exist an index $i$ such that $\alpha_i\geq \alpha_{i+1}$ and the first claim holds.

   Next we turn to the second claim. If $\alpha_i=1$, it must be $B_i=V_i=C_i$, because $V_i$ is the maximal one with $\alpha$-ratio of 1. Therefore, $G_{i+1}=\emptyset$ which means $i=k$. For the case $\alpha_i<1$, if $B_i$ is not independent, then $B_i\cap C_i\neq \emptyset$. Let us consider subset $\widetilde{B}=B_i-(B_i\cap C_i)$. Obviously $\widetilde{B}\neq \emptyset$, otherwise $B_i\subseteq C_i$ and $\alpha_i\geq 1$. For any vertex $v\in \widetilde{B}$, its neighbors in $G_i$ must be contained in $C_i-B_i$, i.e., $\Gamma(\widetilde{B})\cap V_i\subseteq C_i-B_i$. So $\frac{w(\Gamma(\widetilde{B})\cap V_i)}{w(\widetilde{B})}\leq\frac{w(C_i)-w(B_i\cap C_i)}{w(B_i)-w(B_i\cap C_i)}<\alpha_i$, contradicting to the minimality of $\alpha_i$.
 $\square$
 \end{proof}

 \begin{proposition}
   Given the bottleneck decomposition $B=\{(B_1,C_1),\cdots,(B_k,C_k))\}$. If the price to each vertex is set as: for $u\in B_i$, let $p_u=\alpha_iw_u$; and for $u\in C_i$, let $p_u=w_u$, then $(p,X)$ is a market equilibrium, where $X$ is the allocation from BD Mechanism. Furthermore, each agent $u$'s utility is $U_u=w_u\cdot \alpha_i$, if $u\in B_i$; otherwise $U_u=w_u/ \alpha_i$.
 \end{proposition}
 \begin{proof}
   From BD Mechanism, it is not hard to see all resource of each agent is allocated out, which promises the condition of market clearance. In addition, BD Mechanism assigns all resource of each agent to its neighbors from the same pair, that is all available resource are exchanged along edges in $B_i\times C_i$, $i=1,\cdots,k$. More specifically, for any $u\in B_i$, $x_{uv}w_u=f_{uv}$ and $\sum_{v\in \Gamma(u)\cap C_i}f_{uv}=w_u$; and for any $v\in C_i$, $x_{vu}w_v=\alpha_if_{uv}$ and $\sum_{u\in \Gamma(v)\cap B_i}f_{uv}=\frac{w_v}{\alpha_i}$, where $\{f_{uv}\}$ is the maximum flow in the network constructed from bipartite graph of $\hat{G}=(B_i,C_i;E_i)$.
   Now let us turn to the budget constraints based on the previous equalities. For each agent $u\in B_i$,
   \begin{eqnarray}\label{eqn5}
     \sum_{v\in \Gamma(u)\cap C_i}x_{vu}p_v=\sum_{v\in \Gamma(u)\cap C_i}x_{vu}w_v=\alpha_i\sum_{v\in \Gamma(u)\cap C_i}f_{uv}=\alpha_i w_u=p_u;
   \end{eqnarray}
   and for each agent $v\in C_i$,
     \begin{eqnarray}\label{eqn6}
     \sum_{u\in \Gamma(v)\cap B_i}x_{uv}p_u=\alpha_i\sum_{u\in \Gamma(v)\cap B_i}x_{uv}w_u=\alpha_i\sum_{u\in \Gamma(v)\cap B_i}f_{uv}= \alpha_i\cdot\frac{w_v}{\alpha_i}=w_v=p_v.
   \end{eqnarray}
 Thus the allocation from BD mechanism satisfies the budget constraints.

 For the individual optimality, we shall discuss the optimization problems for each agent that $\max\sum_{v\in \Gamma(u)}x_{vu}w_v$, subject to
 $\sum_{v\in \Gamma(u)}x_{vu}p_v\leq p_u$. From the budget constraint, the following inequalities can be deduced directly if the prices are defined as $p_u=\alpha_iw_u$ for $u\in B_i$ and $p_v=w_v$ for $v\in C_i$, that is for each $u\in B_i$,
 \begin{eqnarray*}
  \sum_{v\in \Gamma(u)\cap C_i}x_{vu}p_v=\sum_{v\in \Gamma(u)\cap C_i}x_{vu}w_v=U_u\leq p_u=\alpha_iw_u;
 \end{eqnarray*}
 and for each $v\in C_i$,
 \begin{eqnarray*}
     \sum_{u\in \Gamma(v)\cap B_i}x_{uv}p_u=\alpha_i\sum_{u\in \Gamma(v)\cap B_i}x_{uv}w_u=\alpha_iU_v\leq p_v=w_v.
   \end{eqnarray*}
 Obviously, above two inequalities indicate the upper bound of agent's utility, i.e. $U_u\leq  \alpha_iw_u$ if $u\in B_i$, and $U_v\leq \frac{w_v}{\alpha_i}$ if $v\in C_i$. In addition, the deduction of (\ref{eqn5}) and (\ref{eqn6}) shows each agent's utility reaches the upper bound following the allocation from BD Mechanism. Therefore the condition of individual optimality is satisfied. Of course we can get $U_u=w_u\cdot \alpha_i$, if $u\in B_i$ and $U_u=w_u/ \alpha_i$, if $u\in C_i$.
 $\square$
 \end{proof}

\end{document}